\documentclass[11pt,letterpaper,twoside]{article}

\usepackage{amsmath,amsfonts,amsthm}

\def\l1{{\lambda}_1}

\newcommand{\f}{\frac}

\def\x1{{\xi }_{xx}}
\def\x2{{\xi }_{yy}}
\def\x3{{\xi }_{xy}}

\def\e1{{\eta }_{xx}}
\def\e2{{\eta }_{yy}}
\def\e3{{\eta }_{xy}}

\def\sign{\text{sgn}}

\newcommand{\ds}{\displaystyle }

\newtheorem{theorem}{Theorem}

\newtheorem{corollary}{Corollary}
\newcommand{\beqn}{\begin{eqnarray*}}
\newcommand{\eeqn}{\end{eqnarray*}}
\newcommand{\beqnn}{\begin{eqnarray}}
\newcommand{\eeqnn}{\end{eqnarray}}
\newcommand{\p}{\partial}
\newcommand{\bb}{\begin{equation}}
\newcommand{\ee}{\end{equation}}
\newcommand{\ba}{\begin{array}}
\newcommand{\ea}{\end{array}}
\newcommand{\R}{\mathbb{R}}

\begin{document}
\pagenumbering{arabic}
\title{\huge \bf Peakon and kink solutions for a system of $0-$Holm-Staley equations}
\author{\rm \large Priscila Leal da Silva$^{1,*}$ and Igor Leite Freire$^{2,**}$ \\
\\
\it $1$. Departamento de Matem\'atica, UFSCar, Brazil\\
\it $2$. Centro de Matem\'atica, Computa\c c\~ao e Cogni\c c\~ao, UFABC, Brazil\\
\rm $*$ pri.leal.silva@gmail.com\\
\rm $**$ igor.freire@ufabc.edu.br/igor.leite.freire@gmail.com}
\date{\ }
\maketitle
\vspace{1cm}
\begin{abstract}
In this paper we consider a two-component system of the Holm-Staley equation with no stretching and a one-parameter nonlinearity in the convection term. Point symmetries are found, conditions for the existence of multipeakon and multikink solutions are determined. Solutions for the 1-peakon and 1-kink are explicitly obtained and they exhibit a behaviour quite different of their analogous in the scalar equation.
\end{abstract}
\vskip 1cm
\begin{center}
{2010 AMS Mathematics Classification numbers:\vspace{0.2cm}\\
 35D30, 37K05\vspace{0.2cm}  }
\end{center}
\pagenumbering{arabic}
\newpage

\section{Introduction}

In the work \cite{DHH}, the authors presented the $1$-parameter family of equations
$m_t+um_x=bmu_x,\,m=u-u_{xx}$, today known as the $b$-equation. It is also refereed as Holm-Staley (HS) equation and HS $b$-family of equations as well due to the work of Holm and Staley \cite{holm-staley,holm-staleypla}. The terms $u\,m_x$ and $u_x\,m$ in the HS equation correspond, respectively, to convection and stretching, see \cite{holm-staley} for further details. In a more recent work \cite{Qiao}, the authors deduced a weak multikink solution for the case $b=0$, that is, when the effects of stretching are not considered. It is convenient to name this very particular case and therefore we shall refer to it by the suggestive name $0-$HS equation. Based on these works, in this paper we consider the coupled system, which we accordingly shall name $0-$HS system,
\bb\label{1.2}
m_t=v^bm_x, \quad
n_t =u^bn_x,
\ee
where $b$ is a positive integer, $m=u-u_{xx},\,n=v-v_{xx}$, and deduce peakon and kink solutions for it as done in section \ref{discussion}. Before, in section \ref{main} we prove the main results of the paper.

\section{Main Results}\label{main}

We recall that a Lie point symmetry of the system (\ref{1.2}) is a one-parameter local group of diffeomorphisms preserving the solutions of the system. Then, our first result can be formulated as the following:
\begin{theorem}\label{teo1}
A basis to the Lie point symmetry generators of $(\ref{1.2})$ is given by
\bb\label{2.1}
X_1=\f{\p}{\p x},\,\,X_2=\f{\p}{\p t},\,\,D_b=-bt\f{\p}{\p t}+u\f{\p}{\p u}+v\f{\p}{\p v}
\ee
for any $b$. If $b=1$, in addition to $(\ref{2.1})$ one has a fourth generator given by
\bb\label{2.2}
X_3=-t\f{\p}{\p x}+\f{\p}{\p u}+\f{\p}{\p v}.
\ee
\end{theorem}

\begin{proof}
It follows immediately from the invariance condition, see \cite{bk, ol} and can be directed computed from well known computational routines such as \cite{stelios1,stelios2}.
\end{proof}

Let $\phi$ be a continuous and real valued function defined on an interval $I$. One says that $\phi$ has a {\it peak} at a point $x_0\in I$ if: (1) $\phi$ is smooth on both $I\cap\{z\in\R;\,z<x_0\}$ and $I\cap\{z\in\R;\,z>x_0\}$; (2) the following relations hold $0\neq\lim\limits_{\epsilon\rightarrow 0^+}\phi'(x_0+\epsilon)=-\lim\limits_{\epsilon\rightarrow 0^+}\phi'(x_0-\epsilon)\neq\pm\infty.$

A continuous function $u(x,t)$ is said to have a peak at a point $(x_0,t_0)$ if at least one of the functions $x\mapsto u(\cdot,t_0)$ or $t\mapsto u(x_0,\cdot)$ has a peak at $x_0$ or $t_0$, respectively. For further details, see \cite{pridcds,lenells05}. Finally, we say that a (system of) differential equation has peakon solutions if it admits a solution having a peak in the sense of what has just been discussed. A (system of) differential equation admits a multipeakon solution if it has a solution which is a linear combination of functions having a peak. 

From now on, the functions $p,\,p_i,\,P,\,P_j,\,q,\,q_i,\,Q$ and $Q_j$, where $i=1,\dots,N$ and $j=1,\dots,M$, are smooth functions depending on $t$ and we assume that $(p,q)\not\equiv (0,0)$ and $(P,Q)\not\equiv (0,0)$ as well. For sake of simplicity, we do not write its dependence with respect to the independent variable $t$.

\begin{theorem}\label{teo2}
The pair $(u,v)$ of functions given by 
\bb\label{2.3}
u(x,t) = \sum\limits_{i=1}^N p_ie^{-|x-q_i|}, \quad v(x,t) = \sum\limits_{i=1}^M P_ie^{-|x-Q_i|},
\ee
is a multipeakon solution of the system $(\ref{1.2})$ if and only if the dynamical system
\begin{align}\label{2.4}
p_i' &= bp_i\left(\sum\limits_{j=1}^M P_je^{-|q_i-Q_j|}\right)^{b-1}\left(\sum\limits_{k=1}^MP_k\sign(q_i-Q_k)e^{-|q_i-Q_k|}\right),\nonumber\\
P_i' &= bP_i\left(\sum\limits_{j=1}^N p_je^{-|Q_i-q_j|}\right)^{b-1}\left(\sum\limits_{k=1}^Np_k\sign(Q_i-q_k)e^{-|Q_i-q_k|}\right),\\
q_i' &= -\left(\sum\limits_{j=1}^MP_je^{-|q_i-Q_j|}\right)^b, \quad Q_i' = -\left(\sum\limits_{j=1}^Np_je^{-|Q_i-q_j|}\right)^b\nonumber
\end{align}
is satisfied.
\end{theorem}

\begin{proof}
Let us consider $u$ and $v$ as in (\ref{2.3}). Then, taking their distributional derivatives, we conclude that
\begin{equation}\label{2.5}
m = \sum\limits_{i=1}^Np_i\delta(x-q_i),\quad n =\sum\limits_{i=1}^MP_i\delta(x-Q_i).
\end{equation}
Substituting (\ref{2.3}) and (\ref{2.5}) into (\ref{1.2}) we obtain (\ref{2.4}). Conversely, supposing that $p_i$, $P_i$, $q_i$ and $Q_i$ are given by (\ref{2.4}) and $(u,v)$ by (\ref{2.3}), a straightforward calculation of its weak derivatives shows that such pair is a solution of (\ref{1.2}) in the distributional sense.
\end{proof}

\begin{corollary}\label{cor1}
If $u(x,t)=pe^{-|x-q|}$, $v(x,t)=Pe^{-|x-Q|}$ and if $|\sign{(q-Q)}|=1$, then
$p^b+P^b=\kappa,\quad q'=\left(\kappa/p^b-1\right)Q',$
where $\kappa$ is a constant. On the other hand, if $q=Q$, then $p=c_1$, $P=c_2$ and $c_1$ and $c_2$ are constants satisfying the constraint $c_1^b=c_2^b$.
\end{corollary}

\begin{proof}
In this case, system (\ref{2.4}) is reduced to 
\begin{align}
p'=b\,\sign{(q-Q)}p\,P^b\,e^{-b|q-Q|},\,\,P'=b\,\sign{(Q-q)}P\,p^b\,e^{-b|Q-q|},\nonumber\\
q'=-P^b\,e^{-b|q-Q|},\,\,Q'=-p^b\,e^{-b|Q-q|}.\nonumber
\end{align}
It follows from the fact that 
$p'/P'=-pq'/PQ'=-P^{b-1}/p^{b-1}$, if $q\neq Q$, for all $t$, or, in the case $q=Q$, from the fact that both $p$ and $P$ are constants.
\end{proof}

A kink is a monotonic and bounded function. One says that a differential equation has a kink solution if it has a solution that is a kink. As in the multipeakon case, we say that an equation admits a multikink solution if it has a solution that is a linear combination of kink functions.
\begin{theorem}\label{teo3}
The pair $(u,v)$ given as
\bb\label{2.6}
u =\sum\limits_{i=1}^Nc_i \sign (x-p_i)(e^{-|x-p_i|}-1),\quad v =\sum\limits_{i=1}^M\tilde{c}_i \sign (x-q_i)(e^{-|x-q_i|}-1),
\ee
where $(c_1,\dots,c_N)$ and $(\tilde{c}_1,\dots,\tilde{c}_M)$ are constants,
is a multikink solution of the system $(\ref{1.2})$ if, and only if, the functions $p_i$ and $q_i$ satisfy the system
\begin{align}\label{2.7}
p_i' = \ds{-\left(\sum\limits_{j=1}^M\tilde{c}_j\sign(p_i-q_j)(e^{-|p_i-q_j|}-1)\right)^b},\quad i=1,\dots, N, \\
q_i' = \ds{-\left(\sum\limits_{j=1}^Nc_j\sign(q_i-p_j)(e^{-|q_i-p_j|}-1)\right)^b},\quad i=1,\dots, M.\nonumber
\end{align}
\end{theorem}

\begin{proof}
As in the multipeakon case, calculating the weak derivatives of (\ref{2.6}) we obtain 
$m=-\sum\limits_{i=1}^Nc_i\sign(x-p_i)$ and $n=-\sum\limits_{i=1}^M\tilde{c}_i\sign(x-q_i)$.  Substituting these expressions into (\ref{1.2}), we conclude that it is a solution if (\ref{2.7}) holds. Vice-versa, supposing that $p_i$ and $q_i$ are solutions of (\ref{2.7}) and considering the linear combination (\ref{2.6}), then we conclude that it is a solution of (\ref{1.2}).
\end{proof}

\begin{corollary}\label{cor2}
If $u(x,t)=a_1\,\sign{(x-p)}(e^{-|x-p|}-1)$ and $v(x,t)=a_2\,\sign{(x-q)}\,(e^{-|x-q|}-1)$, with constants $a_1a_2\neq0$, and $A_b:=[(-1)^ba_2^b-a_1^b]/a_1^b$, then
\bb\label{auxk}
(p-q)'=(-1)^{b+1}\,a_1^b\,A_b\,\sign{(p-q)}^b\,(e^{-|p-q|}-1)^b.
\ee
\end{corollary}

\begin{proof}
Similarly as in the Corollary \ref{cor1}, system (\ref{2.7}) becomes 
\begin{align}\label{pqder}
p'=-a_2^b\,\left(\sign{(p-q)}\,(e^{-|p-q|}-1)\right)^b,\,\,q'=(-1)^{b+1}a_1^b\,\left(\sign{(q-p)}\,(e^{-|q-p|}-1)\right)^b.
\end{align}
The result follows after little effort from the fact that
$
(p-q)'=[-a_2^b+(-1)^b\,a_1^b]\,\left(\sign{(p-q)}\,(e^{-|p-q|}-1)\right)^b.
$
\end{proof}

\begin{corollary}\label{cor3}
If $a_1^b=(-1)^b\,a_2^b$, then the pair
\bb\label{pqkink}
p=-a_2^b\,(\sign{(x_0)}\,(e^{-|x_0|}-1))^b\,t+x_0,\quad q=-a_2^b\,(\sign{(x_0)}\,(e^{-|x_0|}-1))^b\,t
\ee
is a solution of the system $(\ref{pqder})$, where $x_0\in\mathbb{R}$.
\end{corollary}
\begin{proof}
It follows from the hypothesis that $A_b=0$. Then, equation (\ref{auxk}) implies that $p=q+x_0$, where $x_0$ is a constant of integration. Solving the system (\ref{pqder}) we obtain the desired result.
\end{proof}
\section{Discussion and Conclusion}\label{discussion}

In Theorem \ref{teo1} we carried out the group classification of system (\ref{1.2}). With exception of the case $b=1$, the symmetries of (\ref{1.2}) are reduced to translations in $x$, $t$ and the scaling $(x,t,u,v)\mapsto (x,e^{-\epsilon b}t,e^{\epsilon}u,e^{\epsilon}v)$. Likewise the classification given in \cite{anco} (see Proposition 1.1), case $b=1$ is somewhat special when the symmetries are taken into account. 

It is worth emphasising that system (\ref{1.2}) also admits other symmetries such as $(x,t,u,v)\mapsto(-x,-t,u,v)$, $(x,t,u,v)\mapsto(x,t,v,u)$, which reflects the fact that it is symmetric with respect to the dependent variables, $(x,t,u,v)\mapsto(x,t,-u,-v)$, if $b$ is even, and $(x,t,u,v)\mapsto(x,-t,-u,-v)$ provided that $b$ is odd.

Some 1-peakon solutions can be obtained by using Corollary \ref{cor1} and the symmetries mentioned in the two previous paragraphs. For instance, if we assume that $q=Q$ and $b$ is even, then we have the solutions $u=-v=c^{1/b}e^{-|x-ct+x_0|}$ or $u=v=c^{1/b}e^{-|x-ct+x_0|}$. The latter pair of functions also provides a solution in the case $b=1,\,3,\,5\cdots$.

With respect to kinks, observe that if $p=q$ in Corollary \ref{cor2}, then $p$ and $q$ are constants, which means that the corresponding solution of system (\ref{1.2}) is
\begin{align}\label{p4.5}
u = c\,\sign(x-k)(e^{-|x-k|}-1), \quad v = \tilde{c}\,\sign(x-k)(e^{-|x-k|}-1),
\end{align}
which is stationary. This remark shows that the coupled system (\ref{1.2}) maintains the property of having stationary 1-kink solutions that was already known for the $0-$HS equation,
see \cite{Qiao}. Actually, in the same work, the authors showed that the equation would only admit a weak 1-kink of the form (\ref{2.6}) if it is stationary. We, however, present a dramatic difference between the $0-$HS equation and the $0-$HS system (\ref{1.2}).  We shall exhibit a non-stationary weak 1-kink solution for the system. 

Corollary \ref{cor3} provides an immediate solution of the 1-kink case, verily,
\bb\label{3.2}
u(x,t)=a_1\,\sign{(x-p)}(e^{-|x-p|}-1),\quad v(x,t)=a_2\,\sign{(x-q)}(e^{-|x-q|}-1),
\ee
where $p$ and $q$ are given by (\ref{pqkink}). Observe that if we choose $x_0=0$ in (\ref{pqkink}) we would then obtain $p=q=0$, which would provide again a stationary solution, a result totally in agreement with that found in \cite{Qiao}. On the other hand, if we chose $x_0\neq0$, from (\ref{pqkink}) we conclude that $p$ and $q$ depends on $t$ and, consequently, (\ref{3.2}) are 1-kink solutions of the $0-$HS system depending on $t$. 

System (\ref{1.2}) is a sort of multicomponent extension of the equations considered in \cite{anco,nois1,nois2} when the ``corresponding stretching'' is neglected. As a consequence of this assumption, and similarly to \cite{Qiao}, we note the emergence of weak kink solutions, although there are different behaviours of the 1-kink solutions of the system when compared to those of the scalar equation considered in \cite{Qiao}. Similarly to what happened in \cite{igordiego} and references thereof, in our case we see that the system seems to have more interesting and challenging properties than the scalar equation. This shows that multicomponent generalisations of the equations considered in \cite{holm-staley} may have fascinating and unexpected properties to be analysed.

\section*{Acknowledgements}

P. L. da Silva is thankful to CAPES for her PhD and post-doc scholarship. I. L. Freire's work is partially supported by CNPq (grant no 308941/2013-6).


\begin{thebibliography}{99}

\bibitem{anco} S. Anco, P. L. da Silva and I. L. Freire, \newblock A family of wave-breaking equations generalizing the Camassa-Holm and Novikov equations, \newblock \emph{J. Math. Phys.}, {\bf 56}, paper 091506, (2015).

\bibitem{bk} G. W. Bluman  and S. Kumei, Symmetries and Differential Equations,  Applied Mathematical Sciences 81, Springer, New York, (1989).

\bibitem{igordiego} D. C. Ferraioli and I. L. Freire,  A generalised multicomponent system of Camassa-Holm-Novikov equations, arXiv:1608.04604v1, (2016).

\bibitem{nois1} P. L. da Silva and I. L. Freire, \newblock Strict self-adjointness and shallow water models, arXiv:1312.3992 (2013). 
 
\bibitem{nois2}  P. L. da Silva and I. L. Freire, \newblock An equation unifying both Camassa-Holm and Novikov equations,  \newblock \emph{Proceedings of the 10th AIMS International Conference}, (2015), DOI: 10.3934/proc.2015.0304.  

\bibitem{pridcds} P. L. da Silva and I. L. Freire, A family of homogeneous peakon equations, (2016), submitted.

\bibitem{DHH} A. Degasperis, D. D. Holm and A. N. W. Hone, \newblock A new integrable equation with peakon solutions, \newblock \emph{Theor. Math. Phys.}, {\bf 133}, 1463--1474, (2002), DOI:10.1016/S0031-8914(53)80099-6.

\bibitem{stelios1} S. Dimas and D. Tsoubelis,  SYM: A new symmetry-finding package for Mathematica. {\em Proceedings of the 10th International Conference in Modern Group Analysis}, Larnaca, Cyprus, 64--70, (2004)

\bibitem{stelios2} S. Dimas and D. Tsoubelis, A new heuristic algorithm for solving overdetermined systems of PDEs in Mathematica. {\em 6th International Conference on Symmetry in Nonlinear Mathematical Physics}, Kiev, Ukraine, 20--26, (2005).

\bibitem{holm-staley} D. D. Holm and M. F. Staley, \newblock Wave structure and nonlinear balances in a family of evolutionary PDEs, \newblock \emph{Siam. J. Appl. Dyn. Sys.}, {\bf 2}, 323--380, (2003).

\bibitem{holm-staleypla} D. Holm and M. Staley, \newblock Nonlinear balance and exchange of stability in dynamics of solitons, peakons, ramp/cliffs and leftons in 1+1 nonlinear evolutionary PDE, \newblock \emph{Phys. Lett. A}, {\bf 308}, 437--444, (2003).

\bibitem{lenells05} J. Lenells, \newblock Traveling wave solutions of the Camassa-Holm equation,\newblock \emph{J. Diff. Equ.}, {\bf 217}, 393--430, (2005).

\bibitem{ol}  P. J. Olver, \newblock Applications of Lie groups to differential equations, 2nd edition, Springer, New York, (1993).

\bibitem{Qiao}B. Xia, Z. Qiao, The n-kink, bell-shape and hat-shape solitary solutions of b-family equation in the case of b=0, \emph{Phys. Lett. A}, {\bf 377}, 2340--2342, (2013), DOI:10.1016/j.physleta.2013.07.017.

\end{thebibliography}
\end{document}